\documentclass[a4paper]{article}
\usepackage[dvipdfmx]{graphicx}
\usepackage[dvipdfmx]{color}

\usepackage{rotating}
\usepackage{comment}
\usepackage{multirow}
\usepackage{hyperref}

\usepackage{latexsym}
\usepackage{amsthm}
\usepackage{amsmath}
\usepackage{amssymb}
\usepackage{url}
\usepackage{cite}

\newcommand{\Suffix}{\mathit{Suffix}}

\newcommand{\STrie}{\mathit{STrie}}
\newcommand{\STree}{\mathit{STree}}
\newcommand{\LSTrie}{\mathit{LSTrie}}

\newcommand{\CDAWG}{\mathit{CDAWG}}
\newcommand{\LCDAWG}{\mathit{L}\mbox{-}\mathit{CDAWG}}
\newcommand{\slen}{\id{slen}}

\newcommand{\esuf}{\id{e\mbox{-}suf}}

\newcommand{\lab}{\id{lab}}
\newcommand{\slink}{\suflink}
\newcommand{\suflink}{\mathit{slink}}

\newcommand{\lspacer}{\kern.2ex}
\newcommand{\spacer}{\rule{0pt}{1.6ex}\kern.2ex}
\newcommand{\rev}[1]{\overline{#1}}

\renewcommand{\subsubsection}[1]{\noindent \textbf{#1}}

\newcommand{\sig}[1]{{\cal #1}}    

\newcommand{\name}{\emph}

\newcommand{\id}[1]{\mathit{#1}}

\newcommand{\polylog}{\mathrm{polylog}}

\newcommand{\Leqr}{\equiv_\mathit{L}}
\newcommand{\Reqr}{\equiv_\mathit{R}}
\newcommand{\Beqr}{\equiv}

\newcommand{\Leqc}[1]{[{#1}]_{L}}
\newcommand{\Reqc}[1]{[{#1}]_{R}}
\newcommand{\Beqc}[1]{[{#1}]}

\newcommand{\Lrep}[1]{\overrightarrow{#1}}
\newcommand{\Rrep}[1]{\overleftarrow{#1}}
\newcommand{\Brep}[1]{\overleftrightarrow{#1}}

\newcommand{\occ}{\mathit{occ}}

\newcommand{\mysubsubsection}[1]{\textbf{#1}:}

\newcommand{\co}[1]{{\tt #1}}      

\newtheorem{theorem}{Theorem}
\newtheorem{lemma}{Lemma}
\newtheorem{proposition}[theorem]{Proposition}
\newtheorem{corollary}[theorem]{Corollary}
\newtheorem{definition}{Definition}

\begin{document}

\title{
  Linear-size CDAWG: new repetition-aware indexing and grammar compression
}

\author{
  Takuya~Takagi$^1$\quad
  Keisuke Goto$^2$\quad
  Yuta Fujishige$^3$\\
  Shunsuke Inenaga$^3$\quad
  Hiroki Arimura$^1$\\
  {$^1$ Graduate School of IST, Hokkaido University, Japan}\\
  {\texttt{\{tkg,arim\}ist.hokudai.ac.jp}}\\
  {$^2$ Fujitsu Laboratories Ltd., Japan}\\
  {\texttt{goto.keisuke@jp.fujitsu.com}}\\
  {$^3$ Department of Informatics, Kyushu University, Japan}\\
  {\texttt{\{yuta.fujishige,inenaga\}@inf.kyushu-u.ac.jp}}\\
}

\maketitle

\pagestyle{plain}

\begin{abstract}
In this paper, we propose a novel approach to combine
\emph{compact directed acyclic word graphs} (\emph{CDAWGs})
and grammar-based compression.
This leads us to an efficient self-index,
called \emph{Linear-size CDAWGs} (\emph{L-CDAWGs}),
which can be represented with
$O(\tilde e_T \log n)$
bits of space allowing for
$O(\log n)$-time random and 
$O(1)$-time sequential accesses to edge labels, 
and $O(m \log \sigma + occ)$-time pattern matching.
Here,
$\tilde e_T$ is the number of all extensions of maximal repeats in $T$,
$n$ and $m$ are respectively the lengths of the text $T$ and a given pattern,
$\sigma$ is the alphabet size,
and $\occ$ is the number of occurrences of the pattern in $T$.
The repetitiveness measure $\tilde e_T$ is 
known to be much smaller than the text length $n$
for highly repetitive text.
For constant alphabets, our L-CDAWGs achieve
$O(m + \occ)$ pattern matching time with $O(e_T^r \log n)$ bits of space,
which improves the pattern matching time of
Belazzougui et al.'s run-length BWT-CDAWGs by a factor of $\log \log n$,
with the same space complexity. 
Here, $e_T^r$ is the number of right extensions of maximal repeats in $T$. 
As a byproduct, our result gives a way of constructing a straight-line program (SLP) of size $O(\tilde e_T)$ for a given text $T$ in $O(n + \tilde e_T \log \sigma)$ time. 


\end{abstract}


\section{Introduction}

\mysubsubsection{Background}
Text indexing is a fundamental problem in theoretical computer science,
where the task is to preprocess a given text so that
subsequent pattern matching queries can be answered quickly.
It has wide applications such as information retrieval, bioinformatics, and big data analytics~\cite{navarro:makinen:2007compressed,kreft:navarro:2013}.
There have been a lot of recent research on \name{compressed text indexes}~%
\cite{%
  ClaudeN11,%
  kreft:navarro:2013,%
  navarro:makinen:2007compressed,%
  belazzougui:BWTCDAWG:2015,%
  navarro2016self:bt,%
  makinen:navarro:siren:valimaki:2010bwt,%
  takabatake:tabei:sakamoto:2014imp:esp,%
  grossi2003high
}
that store a text $T$ supporting \co{extract} and \co{find} operations in space significantly smaller  than the total size $n$ of texts. 
Operation \co{extract} returns any substring $T[i..j]$ of the text.
Operation \co{find} returns the list of all $occ$ occurrences of a given pattern $P$ in $T$.
For instance, Grossi, Gupta, and Vitter~\cite{grossi2003high} gave a compressed text index based on compressed suffix arrays, which takes $s = n H_k + O(n \log\log n\log\sigma/\log n)$ bits of space and supporting $O(m \log\sigma + \polylog(n))$ pattern match time, where $H_k$ is the $k$-th order entropy of $T$
and $m$ is the length of the pattern $P$.

\mysubsubsection{Compression measures for highly repetitive text}
Recently, there has been an increasing interest in indexed searches for highly repetitive text collections. Typically, the compression size of such a text can be described in terms of some measure of repetition. The followings are examples of such repetitiveness measures for $T$:
\begin{itemize}
\item the \name{number $g_T$ of rules in a grammar (SLP)} representing $T$, 
\item the \name{number $z_T$ of phrases in the LZ77 parsing} of $T$, 
\item the \name{number $r_T$ of runs in the Burrows-Wheeler transform} of $T$, and 
\item the \name{number $\tilde e_T = e_T^r + e_T^\ell$ of right- and left-extensions of maximal repeats} of~$T$.
\end{itemize}
Belazzougui \emph{et al.}~\cite{belazzougui:BWTCDAWG:2015} observed close relationship among these measures. Specifically, the authors empirically observed that all of them showed similar logarithmic growth behavior in $|T|$ on a real biological sequence, and also theoretically showed that both $z_T$ and $r_T$ are upper bounded by $\tilde e_T$.
These repetitive texts are formed from many repeated fragments nearly identical. Therefore, one can expect that compressed index based on these measures such as $g_T, z_T, r_T$, and $\tilde e_T$ can effectively capture the redundancy inherent to these highly repetitive texts than conventional entropy-based compressed indexes~\cite{navarro:makinen:2007compressed}.

\mysubsubsection{Repetition-aware indexes}
There has been extensive research on a family of \name{repetition-aware indexes}~\cite{belazzougui:BWTCDAWG:2015,ClaudeN11,kreft:navarro:2013,makinen:navarro:siren:valimaki:2010bwt} since the seminal work by Claude and Navarro~\cite{ClaudeN11}. 
They proposed 
the first
compressed self-index based on grammars, which takes $s = g \log n + O(g \log g)$ bits supporting $O((m^2 + h(m + \occ))\log g)$ pattern match time, where $g = g_T$ and $h$ are respectively the size and height of a grammar.
Kreft and Navarro~\cite{kreft:navarro:2013} gave
the first compressed self-index based on LZ77, which takes $s = 3z \log n + 5n \log \sigma + O(z) + o(n)$ bits 
supporting $O(m^2d + (m + occ)\log z)$ pattern match time.
Here, $d$ is the height of the LZ parsing. 
Makinen, Navarro, Siren, and Valimaki~\cite{makinen:navarro:siren:valimaki:2010bwt} gave
a compressed index based on RLBWT, which takes $s = r \log\sigma\log(2n/r)(1+o(1)) + O(r\log\sigma\log\log(2n/r)) + O(\sigma\log n)$ bits supporting $O(m\, f(r\log\sigma, n\log\sigma))$ pattern match time, 
where 
$f(b,u)$ is the time for a binary searchable dictionary which is $O((\log b)^{0.5})$ and $o((\log\log u)^2)$ for example~\cite{makinen:navarro:siren:valimaki:2010bwt}.


\mysubsubsection{Previous approaches}
Considering the above results, we notice that in compression ratio, all indexes above achieve good performance depending on the repetitive measures, while in terms of operation time, most of them except the RLBWT-based one~\cite{makinen:navarro:siren:valimaki:2010bwt} have quadratic dependency in pattern size $m$.
Hence, a challenge here is to develop repetition-aware text indexes to achieve good compression ratio for highly repetitive texts in terms of repetition measures, while supporting faster \co{extract} and \co{find} operations. 
Belazzougui \emph{et al.}~\cite{belazzougui:BWTCDAWG:2015} proposed a repetition-aware index which combines \emph{CDAWGs}~\cite{blumer:bbhme:1987complete,crochemore1997cdawg} and the run-length encoded BWT~\cite{makinen:navarro:siren:valimaki:2010bwt}, to which we refer as \name{RLBWT-CDAWGs}. For a given text $T$ of the length $n$ and a pattern $P$ of the length $m$,
their index 
uses $O(e_T^r \log n)$ bits of space and
supports \co{find} operation in $O(m \log\log n + \occ)$ time.

\mysubsubsection{Main results}
In this paper, we propose a new repetition-aware index based on combination of CDAWGs and grammar-based compression, called the \emph{Linear-size CDAWG} (L-CDAWG, for short). The L-CDAWG of a text $T$ of length $n$ is a self-index for $T$
which can be stored in $O(\tilde e_T \log n)$ bits of space,
and support $O(\log n)$-time random access to the text,
$O(1)$-time sequential character access from the beginning of each edge label,
and $O(m \log \sigma + occ)$-time pattern matching.
For constant alphabets,
our L-CDAWGs use $O(e_T^r \log n)$ bits of space
and support pattern matching in $O(m + \occ)$ time,
hence improving the pattern matching time of Belazzougui \emph{et al.}'s
RLBWT-CDAWGs by a factor of $\log \log n$.
We note that RLBWT-CDAWGs use hashing to retrieve the first character
of a given edge label,
and hence RLBWT-CDAWGs seem to require $O(m \log\log n + \occ)$ time
for pattern matching even for constant alphabets.

From the context of studies on \name{suffix indices}, 
our L-CDAWGs can be seen as a successor of 
the \emph{linear-size suffix trie} (\name{LSTries}) by Crochemore~\emph{et al.}~\cite{epifanio:mignosi:grossi:crochemore:2016lst}.
The LSTrie is a variant of the suffix tree~\cite{crochemore:rytter:2003jewels}, which need not keep the original text $T$ by elegant scheme of linear time decoding using suffix links and a set of auxiliary nodes. 
However, it is a challenge to generalize their result for the CDAWG because the paths between a given pair of endpoints are not unique. 
By combining the idea of LSTries, an SLP-based compression with direct access~\cite{Gasieniec:kolpakov:DCC:2005,Bille:Landau:Raman:Sadakane:Rao:Weimann:2015}, we successfully devise a text index of $O(\tilde e_T\log n)$ bits by improving functionalities of LSTries. As a byproduct, our result gives a way of constructing an SLP of size $O(\tilde e_T\log \tilde e_T)$ bits of space for a text $T$. 
Moreover, since the L-CDAWG of $T$ retains the topology of the original CDAWG for $T$,
the L-CDAWG is a compact representation of all maximal repeats~\cite{raffinot:2001:max:repeats} that appear in $T$.


\section{Preliminaries}
\label{sec:prelim}

In this section, we give some notations and definitions 
to be used in the following sections. 
In addition, we recall string data structures such as
suffix tries,
suffix trees, CDAWGs, linear-size suffix tries
and straight-line programs,
which are the data structures to be considered in this paper.

\subsection{Basic definitions and notations} 

\mysubsubsection{Strings} 
Let $\Sigma$
be a general ordered alphabet of size $\sigma\ge 2$.
An element $T = t_1\cdots t_n$ of $\Sigma^*$ is called a \name{string}, where
$|T| = n$ denotes its length.
We denote the empty string by $\varepsilon$ which is the string of length $0$, namely, $|\varepsilon| = 0$.
Let $\Sigma^+ = \Sigma^* \setminus \{\varepsilon\}$.
If $T = XYZ$,
then $X$, $Y$, and $Z$ are called 
a \emph{prefix}, a \emph{substring}, and a \emph{suffix} of $T$, respectively. 
Let $T= t_1\cdots t_n \in \Sigma^n$ be any string of length $n$. 
For any $1 \leq i \leq j \leq n$, 
let $T[i..j] = t_i\cdots t_j$ denote the substring of $T$ that begins and ends at positions $i$ and $j$ in $T$, 
and let $T[i] = t_i$ denote the $i$th character of $T$. 
For any string $T$,
we denote by $\rev{T}$ the reversed string of $T$,
i.e., $\rev{T} = T[n]\cdots T[1]$. 
Let $\Suffix(T)$ denote the set of suffixes of $T$.
For a string $x$, the number of occurrences of $x$ in $T$ 
means the number of positions where $x$ is a substring in $T$.

\mysubsubsection{Maximal repeats and other measures of repetition}
A substring $w$ of $T$ is called a \emph{repeat} if the number of occurrences of $w$ in $T$ more than one.
A \name{right extension} (resp. a \name{left extension}) of $w$ of $T$ is any substring of $T$ with the form $wa$ (resp. $aw$) for some letter $a \in \Sigma$. 
A repeat $w$ of $T$ is a \emph{maximal repeat} if both left- and right-extensions of $w$ occur strictly fewer times in $T$ than $w$.
In what follows, we denote by
$\mu_T$, $e_T^r$, $e^{\ell}_T$, and $\tilde e_T = e_T^r + e^{\ell}_T$
the numbers of maximal repeats,
right-extensions, left-extensions, and all extensions of maximal repeats appearing in $T$, respectively.
Recently, it has been shown in~\cite{belazzougui:BWTCDAWG:2015} that the number $\tilde e_T$ is an upper bound on the number $r_T$ of runs in the Burrows-Wheeler transform for $T$ and the number $z_T$ of factors in the Lempel-Ziv parsing of $T$.
It is also known that $\tilde e_T \leq 4n-4$ and
$\mu_T < n$, where $n = |T|$~\cite{blumer:bbhme:1987complete,raffinot:2001:max:repeats}.

\mysubsubsection{Notations on graphical indexes}
All index structures dealt with in this paper, such as suffix tries, suffix trees, CDAWGs, linear-size suffix tries (LSTries), and linear-size CDAWGs (L-CDAWGs), are \name{graphical indexes} in the sense that an index is a pointer-based structure built on an underlying DAG $G_L = (V(L), E(L))$ with a root $r \in V(L)$ and mapping $lab: E(L) \to \Sigma^+$ that assign a label $\lab(e)$ to each edge $e \in E(L)$.
For an edge $e = (u, v) \in E(L)$, we denote its \name{end points} by $e.hi := u$ and $e.lo := v$, respectively. The \name{label string} of $e$ is $\lab(e) \in \Sigma^+$. The \name{string length} of $e$ is $\slen(e) := |\lab(e)| \ge 1$. An edge is called \name{atomic} if $\slen(e) = 1$, and thus, $\lab(e) \in \Sigma$. 
For a path $p = (e_1, \ldots, e_k)$ of length $k\ge 1$, we extend its \name{end points}, \name{label string}, and \name{string length} by $p.hi := e_1.hi$, $p.lo := e_k.lo$, $\lab(p) := \lab(e_1)\dots \lab(e_k) \in \Sigma^+$, and $\slen(p) := \slen(e_1)+ \dots + \slen(e_k)\ge 1$, respectively.

\begin{figure}[t]
  \centerline{
    \includegraphics[width=1.0\linewidth]{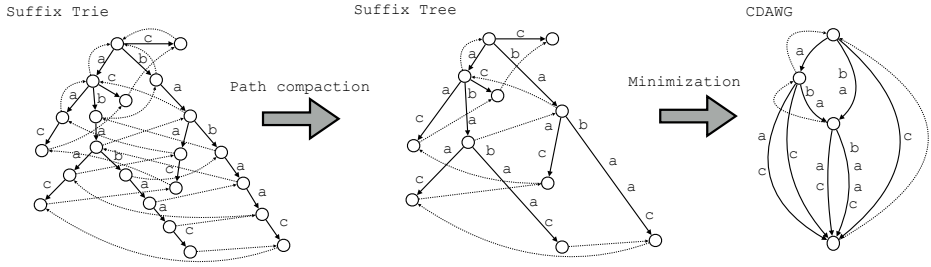}
  }
  \caption{Illustration of $\STrie(T)$, $\STree(T)$,
  and $\CDAWG(T)$ with $T = \mathtt{ababaac}$.
  The solid arrows and broken arrows represent 
  the edges and the suffix links of each data structure, respectively.
  }
    \label{fig:strie_stree_cdawg}
\end{figure}

\subsection{Suffix tries and suffix trees} 

\label{sec:def_ST_DAWG} 
The \emph{suffix trie}~\cite{crochemore:rytter:2003jewels} for a text $T$ of length $n$, 
denoted $\STrie(T)$, is a trie 
which represents $\Suffix(T)$.
The size of $\STrie(T)$ is $O(n^2)$. 
The path label of a node $v$ is the string $str(v) := \lab(\pi_v)$ formed by concatenating the edge labels on the unique path $\pi_v$ from the root to $v$.
If $x = str(v)$, we denote $v$ by $[x]$. We may identify $v = [x]$ with its label $x$ if it is clear from context. 
A substring $x$ of $T$ is said to be \emph{branching} if there exists two distinct characters $a,b \in \Sigma$ such that both $xa$ and $xb$ are substrings of $T$. 
For any $a \in \Sigma$, $x \in \Sigma^*$,
we define the \emph{suffix link} of node $[ax]$ by $\suflink([ax]) = [x]$ if $[ax]$ is defined. 

The \emph{suffix tree}~\cite{Weiner,crochemore:rytter:2003jewels} for a text $T$, 
denoted $\STree(T)$, is a compacted trie which also 
represents $\Suffix(T)$.
$\STree(T)$ can be obtained by compacting every path of $\STrie(\sig{T})$ which consists of non-branching internal nodes (see Fig.~\ref{fig:strie_stree_cdawg}). 
Since every internal node of $\STree(T)$ is branching, and since there are at most $n$ leaves in $\STree(T)$, 
the numbers of edges and nodes are $O(n)$. 
The edges of $\STree(T)$ are labeled by non-empty substrings of $T$. 
By representing each edge label $\alpha$ with a pair $(i, j)$ of integers such that $T[i..j] = \alpha$, $\STree(T)$ can be stored in
$O(n \log n)$ bits of space.

\subsection{CDAWGs}

The \emph{compact directed acyclic word graph}~\cite{blumer:bbhme:1987complete,crochemore:rytter:2003jewels}
for a text $T$, denoted $\CDAWG(T)$,
is the minimal compact automaton which represents $\Suffix(T)$. 
$\CDAWG(T)$ can be obtained from $\STree(T\$)$ by merging isomorphic subtrees and deleting associated endmarker $\$ \not\in \Sigma$. 
Since $\CDAWG(T)$ is an edge-labeled DAG, we represent a directed edge from node $u$ to $v$ with label string $x \in \Sigma^+$ by a triple $f = (u, x, v)$.
For any node $u$,
the label strings of out-going edges from $u$ start with mutually distinct characters. 

Formally, $\CDAWG(T)$ is defined as follows.
For any strings $x, y$,
we denote $x \Leqr y$ (resp. $x \Reqr y$) iff
the beginning positions (resp. ending positions) of $x$ and $y$ in $T$
are equal.
Let $\Leqc{x}$ (resp. $\Reqc{x}$) denote the
equivalence class of strings w.r.t. $\Leqr$ (resp. $\Reqr$).
All strings that are \emph{not} substrings of $T$ form
a single equivalence class, and in the sequel we will consider only 
the substrings of $T$.
Let $\Lrep{x}$ (resp. $\Rrep{x}$) denote
the longest member of the equivalence class $\Leqc{x}$ (resp. $\Reqc{x}$).
Notice that each member of $\Leqc{x}$ (resp. $\Reqc{x}$)
is a prefix of $\Lrep{x}$ (resp. a suffix of $\Rrep{x}$).
Let $\Brep{x} = \Rrep{(\Lrep{x})} = \Lrep{(\Rrep{x})}$.
We denote $x \Beqr y$ iff $\Brep{x} = \Brep{y}$,
and let $\Beqc{x}$ denote the equivalence class w.r.t. $\Beqr$.
The longest member of $\Beqc{x}$ is $\Brep{x}$
and we will also denote it by $value([x])$.
We define $\CDAWG(T)$ as an edge-labeled DAG $(V, E)$ such that
$V = \{\Reqc{\Lrep{x}} \mid \mbox{$x$ is a substring of $T$}\}$ and 
$E = \{(\Reqc{\Lrep{x}}, \alpha, \Reqc{\Lrep{x}\alpha}) \mid \alpha \in \Sigma^+, \Lrep{x} \not \Beqr \Lrep{x}\alpha\}$.
The $\Lrep{\cdot}$ operator corresponds to compacting
non-branching edges (like conversion from $\STrie(T)$ to $\STree(T)$)
and the $\Reqc{\cdot}$ operator corresponds to
merging isomorphic subtrees of $\STree(T)$.
For simplicity,
we abuse notation so that
when we refer to a node of $\CDAWG(T)$ as $[x]$,
this implies $x = \Lrep{x}$ and $[x] = \Reqc{\Lrep{x}}$.

Let $[x]$ be any node of $\CDAWG(T)$ and
consider the suffixes of $value([x])$ 
which correspond to the suffix tree nodes that are merged when transformed into the CDAWG.
We define the \emph{suffix link} of
node $[x]$ by $\suflink([x]) = [y]$, iff
$y$ is the longest suffix of $value([x])$ that does not belong to $[x]$.

It is shown
that all nodes of $\CDAWG(T)$ except the sink
correspond to the maximal repeats of $T$. Actually, $value([x])$ is a maximal repeat in $T$~\cite{raffinot:2001:max:repeats}. 
Following this fact, one can easily see that the numbers of edges of $\CDAWG(T)$ and $\CDAWG(\rev{T})$ coincide with the numbers $e_T^r$ and $e_T^\ell$ of right- and left- extensions of maximal repeats of $T$, respectively~\cite{belazzougui:BWTCDAWG:2015,raffinot:2001:max:repeats}.

By representing each edge label $\alpha$ with pairs $(i, j)$ of integers such that $T[i..j] = \alpha$, $\CDAWG(T)$ can be stored in $O(e_T^r \log n + n \log \sigma)$ bits of space.

\begin{figure}[t]
  \centerline{
    \includegraphics[width=1.0\linewidth]{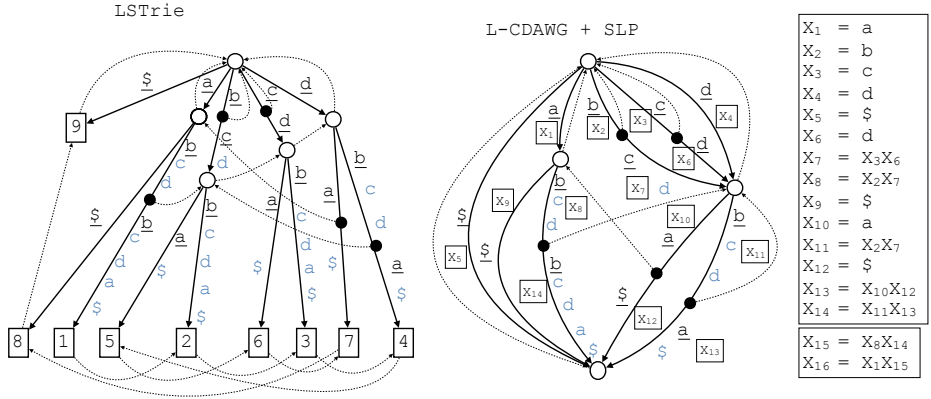}
  }
  \caption{
    Illustration of $\LSTrie(T)$ 
    and our index structure $\LCDAWG(T)$ with $SLP$ 
    for text $T = \mathtt{abcdbcda\$}$.
    Solid and broken arrows represent the edges and suffix links, 
    respectively.
    Underlined and shaded characters attached to each edge are the first (real) and the following (virtual) characters of the original edge label.
    The expression $X_i$ at the edge indicates the $i$-th variable of the SLP for $T$. 
    }
  \label{fig:lstrie_lcdawg}
  \label{fig:slp_lcdawg}
\end{figure}

\subsection{LSTrie}

Recently, Crochemore~\emph{et al.}~\cite{epifanio:mignosi:grossi:crochemore:2016lst} proposed a compact variant of a suffix trie, called \emph{linear-size suffix trie} (or LSTrie, for short), denoted $\LSTrie(T)$. It is a compacted tree with the topology and the size similar to $\STree(T)$, but has no indirect references to a text $T$~(See Fig.~\ref{fig:lstrie_lcdawg}). 
$\LSTrie(T)$ is obtained from $\STree(T)$ by adding all nodes $v$ such that their suffix links $slink(v)$ appear also in $\STree(T)$. Unlike $\STree(T)$, 
each edge $(u,v)$ of $\LSTrie(T)$ stores the first character
and the length of the corresponding suffix tree edge label (see Fig.~\ref{fig:lstrie_lcdawg}).
Using auxiliary links called the \name{jump pointers} the following theorem is proved.

\begin{proposition}[Crochemore~\emph{et al.}~\cite{epifanio:mignosi:grossi:crochemore:2016lst}]
  For a text $T$ of length $n$, the linear-size suffix trie $\LSTrie(T)$ for $T$ can be stored in $O(n \log n)$ bits of space supporting reconstruction of the label of a given edge in $O(\ell)$ time, where $\ell$ is the length of the edge label. 
\end{proposition}

Crochemore~\emph{et al.}'s method~\cite{epifanio:mignosi:grossi:crochemore:2016lst} does not regard the order of decoding characters on an edge label.
This implies that $\LSTrie(T)$ needs $O(\ell)$ worst case time to read any prefix of an edge label of length $\ell$.
This may cause troubles in some applications including pattern matching. 
In particular, it does not seem straightforward to match a pattern $P$ against 
a prefix of the label of an edge $e$ in $O(|P|)$ time when $|P| < |\lab(e)|$. 
We will solve these problems in Section~\ref{sec:algo} later.

\subsection{Straight-line programs}

A straight-line program (SLP) is a context-free grammar (CFG) 
in the Chomsky normal form generating a single string.
SLPs are often used in grammar compression algorithms~\cite{navarro:makinen:2007compressed}.

Consider an SLP $R$ with $n$ variables.
Each production rule is either of form $X \to a$ with $a \in \Sigma$
or $X \to YZ$ without loops.
Thus an SLP produces a single string.
The \emph{phrase} of each $X_i$, denoted $\sig F(X_i)$,
is the string that $X_i$ produces.
The string defined by SLP $R$ is $\sig F(X_n)$.
We will use the following results.


\begin{proposition}[Gasieniec \emph{et al.}~\cite{Gasieniec:kolpakov:DCC:2005}]
  \label{Gasieniec}
  For an SLP $R$ of size $g$ for a text of length $n$,
  there exist a data structure of $O(g \log n)$ bits of space which supports expansion of a prefix of $\sig F(X_i)$ for any variable $X_i$ in $O(1)$ time per character, and can be constructed in $O(g)$ time.

\end{proposition}

\begin{proposition}[Bille \emph{et al.}~\cite{Bille:Landau:Raman:Sadakane:Rao:Weimann:2015}]
  \label{Bille}
  For an SLP $R$ of size $g$ representing a text of length $n$,
  there exists a data structure of $O(g \log n)$ bits of space
  which supports to access consecutive $m$ characters at arbitrary position of $\sig F(X_i)$ for
  any variable $X_i$ in $O(m + \log n)$ time,
  and can be constructed in $O(g)$ time.   
\end{proposition}

\section {The proposed data structure: L-CDAWG}
\label{sec:algo}

In this section, we present \emph{the Linear-size CDAWG} (\emph{L-CDAWG}, for short).
The L-CDAWG can support CDAWG operations in the same time complexity without holding the original input text and  can reduce the space complexity from
$O(e_T^r \log n + n \log \sigma)$ bits of space to $O(\tilde e_T\log n)$ bits of space, 
where $\tilde e_T = e_T^r + e_T^\ell$ is the number of extensions of maximal repeats.
From now on, we assume that an input text $T$ terminates with a unique
character \$ which appears nowhere else in $T$.

\subsection {Outline}
\label{algo:def_lcdawg}

The \emph{Linear-size CDAWG}
for a text $T$ of length $n$, denoted $\LCDAWG(T)$, is a DAG whose edges are labeled with single characters. 
$\LCDAWG(T)$ can be obtained from $\CDAWG(T)$ by the following modifications. 
From now on, we refer to the original nodes appearing in $\CDAWG(T)$ as
\emph{type-1 nodes}, which are always branching except the sink.

\begin{enumerate}
  
\item First, we add new non-branching nodes, called \emph{type-2 nodes} to $\CDAWG(T)$. Let $u = value(\Beqc{x})$ for any type-1 node $\Beqc{x}$ of $\CDAWG(T)$. If $au$ is a substring of $T$ but the path spelling out $au$ ends in the middle of an edge, then we introduce a type-2 node $v$ representing $au$.
We add the suffix link $u = \slink(v)$ as well. Adding type-2 nodes splits an edge into shorter ones. Note that more than one type-2 nodes can be inserted into an edge of $\CDAWG(T)$.
  
\item Let $(u, x, v)$ be any edge after all the type-2 nodes are inserted, where $x \in \Sigma^+$. We represent this edge by $e = (u, c, v)$ where $c$ is the first character $c = x[1] \in \Sigma$ of the original label. We also store the original label length $\slen(e) = |x|$.

\item We will augment $\LCDAWG(T)$ with a set of SLP production rules whose nonterminals correspond to edges of $\LCDAWG(T)$. The definition and construction of this SLP will be described later in Section~\ref{subsec:algo:grammar}. 
\end{enumerate}

If non-branching type-2 nodes are ignored,
then the topology of $\LCDAWG(T)$ is the same as that of $\CDAWG(T)$.
For ease of explanation, we denote by $\lab(e)$ 
the original label of edge $e$.
Namely, for any edge $e = (u, c, v)$,
$\lab(e) = x$ iff $(u, x, v)$ is the original edge for $e$.

The following lemma gives an upper bound of the numbers of nodes and edges in $\LCDAWG(T)$.
Recall that $\mu_T$ is the number
of maximal repeats in $T$,
$e_T^\ell$ and $e_T^r$ are respectively the number of left- and right-extensions
of maximal repeats in $T$,
and $\tilde e_T = e_T^\ell + e_T^r$.

\begin{lemma}
  \label{lemma:LCDAWG_size}
  For any string $T$, let $\LCDAWG(T) = (V, E)$,
  then $|V| = O(\mu_T + e_T^\ell)$ and $|E| = O(\tilde e_T)$.

\end{lemma}

\begin{proof}
  Let $\CDAWG(T) = (V_0, E_0)$ and $\CDAWG(\rev{T}) = (\rev{V_0}, \rev{E_0})$.
  It is known that $|V_0| = |\rev{V_0}| = \mu_T$,
  $|E_0| = e_T^r$ and $|\rev{E_0}| = e_T^\ell$
  (see~\cite{blumer:bbhme:1987complete} and~\cite{raffinot:2001:max:repeats}).
  Let $V_1$ and $V_2$ be the set of type-1 and type-2 nodes in
  $\LCDAWG(T)$, respectively.
  Clearly, $V_1 \cap V_2 = \emptyset$,
  $V = V_1 \cup V_2$, and $V_1 = V_0$.
  Let $\Beqc{x} \in V_1$ and $u = value({\Beqc{x}})$.
  Note that $u$ is a maximal repeat of $T$.
  For any character $a \in \Sigma$ such that $au$ is a substring of $T$,
  clearly $au$ is a left-extension of $u$.
  By the definition of $\LCDAWG(T)$,
  it always has a (type-1 or type-2) node which corresponds to $au$.
  Hence $|V_2| \leq e_T^\ell$.
  This implies $|V| = |V_1| + |V_2| = O(\mu_T + e_T^\ell)$.
  Since each type-2 node is non-branching,
  clearly $|E| = O(e_T^r + e_T^\ell) =O(\tilde e_T)$.
  \qed 
\end{proof}

\begin{corollary}
  For any string of $T$ over a constant alphabet,
  $|V| = O(\mu_T + e_T^r)$ and $|E| = O(e_T^r)$,
  where $\LCDAWG(T) = (V, E)$.
\end{corollary}

\begin{proof}
  It clearly holds that $\mu_T \geq e_T^\ell / \sigma$
  and $e_T^r \geq \mu_T$.
  Thus we have $e_T^\ell \leq \sigma e_T^r$.
  The corollary follows from Lemma~\ref{lemma:LCDAWG_size}
  when $\sigma = O(1)$.
  \qed
\end{proof}

\subsection{Constructing type-2 nodes and edge suffix links}

\begin{lemma} \label{lem:computing_type-2}
  Given $\CDAWG(T)$ for a text $T$,
  we can compute all type-2 nodes of $\LCDAWG(T)$ in
  $O(\tilde e_T \log \sigma)$ time.
\end{lemma}

\begin{proof}
  We create a copy $G$ of $\CDAWG(T)$.
  For each edge $(u, x, v)$ of $\CDAWG(T)$,
  we compute node $u' = \suflink(u)$
  and  the path $Q$ that spells out $x$ from $u'$.
  The number of type-1 nodes in this path $Q$
  is equal to the number of type-2 nodes
  that need to be inserted on edge $(u, x, v)$,
  and hence we insert these nodes to $G$.
  After the above operation is done for all edges,
  $G$ contains all type-2 nodes of $\LCDAWG(T)$.
  Since there always exists such a path $Q$,
  to find $Q$ it suffices to check the first characters
  of out-going edges. Hence we need only $O(\log \sigma)$ time
  for each node in $Q$.
  Overall, it takes $O(\tilde e_T \log \sigma)$ time.
  \qed
\end{proof}


The above lemma also indicates the notion of the following \emph{edge suffix links}
in $\LCDAWG(T)$ which are virtual links, and will not be actually created in the construction. 

\begin{definition}[Edge suffix links]\rm 
  \label{def:e-suf}
  For any edge $e$ with $\slen(e) \ge 2$, $\esuf(e) = (e_1, \ldots, e_k )$ is the path, namely a list of edges, from $e_1.hi=\suflink(e.hi)$ to $e_k.lo$ that can be reachable from $e_1.hi$ by scanning $\lab(e)$.

\end{definition}

Edge suffix links have the following properties.

\begin{lemma} \label{lem:esuf:nonbranch}
  For any edge $e$ such that $\slen(e) \ge 2$ and its edge suffix link $\esuf(e)=(e_1, \ldots, e_k )$, (1) both $e_1.hi$ and $e_k.lo$ are type-1 nodes, and (2)
  all nodes in the path $e_1.lo = e_2.hi, \ldots, e_{k-1}.lo = e_k.hi$ are type-2 nodes.
\end{lemma}

\begin{proof}
  From the definition of edge suffix links,
  we have $e_1.hi = \suflink(e.hi)$ and
  the path from $e_1.hi$ to $e_k.lo$ spells out $\lab(e)$.
  (1) By the definitions of type-2 nodes and
  edge suffix links, $e_1.hi$ is always of type-1.
  Hence it suffices to show that $e_k.lo$ is of type-1.  
  There are two cases:
  (a) If $e.lo$ is a type-2 node,
  then by the definition of type-2 nodes,
  $e_k.lo$ must be the node pointed by $\suflink(e.lo)$.
  Therefore, $e_k.lo$ is a type-1 node.
  (b) If $e.lo$ is a type-1 node,
  then let $ax$ be the shortest string represented by $e.hi$
  with $a \in \Sigma$ and $x \in \Sigma^*$.
  Then, string $x \cdot \lab(e)$ is spelled out by a path
  from the source to $e_1.hi, \ldots, e_k.lo$,
  where either $e_k.lo = e.lo$ or $e_k.lo = \slink(e.lo)$.
  Since $e.lo$ is of type-1, $\slink(e.lo)$ is also of type-1.
  (2) If there is a type-1 node $u$ in the path $e_2.hi, \ldots, e_{k-1}.lo$,
  then there has to be a (type-1 or type-2) node $v$ between
  $e.hi$ and $e.lo$, a contradiction. 
  \qed 
\end{proof}

Lemma~\ref{lem:esuf:nonbranch}
says that the label of any edge $e = (u, c, v)$ with $\slen(e) \ge 2$ can be represented by a path $p = (e_1, \ldots, e_k) = \esuf(e)$. 
In addition, since the path $p$ includes type-1 nodes only at the end points
and since type-2 nodes are non-branching,
$p$ is uniquely determined by a pair of $(\slink(u), c)$.
We can compute all edges $e_i \in p$ for $1 \le i \le k$ 
in $O(k + \log \sigma)$ per query, as follows.
Firstly, we compute $p.hi = \suflink(u)$ and then select the out-going edge $e_1$ starting with the character $c$ in $O(\log \sigma)$ time. 
Next, we blindly scan the downward path from $e_1$ while the lower end of the current edge $e_i$ has type-2. This scanning terminates when we reach an edge $e_k$ such that $e_k.lo$ is of type-1.

\subsection{Construction of the SLP for  L-CDAWG}
\label{subsec:algo:grammar}

We give an SLP of size $O(\tilde e_T)$ which represents $T$ and all edge labels of $L = \LCDAWG(T)$ based on the jump links.

\mysubsubsection{Jumping from an edge to a path}
First, we define \emph{jump links}, by which  we can jump from a given edge $e$ with $\slen(e) \ge 2$ to the path consisting of at least two edges, and having the same string label.
Although our jump link is based on that of LSTries~\cite{epifanio:mignosi:grossi:crochemore:2016lst}, we need a new definition since a path in $\CDAWG(T)$ (and hence in $\LCDAWG(T)$) cannot be uniquely determined by a pair of nodes,
unlike $\STree(T)$ (or $\LSTrie(T)$).

\begin{definition}[Jump links]\rm 
  For an edge $e$ with $\slen(e) \ge 2$ and $\esuf(e)=(e_1, \ldots, e_k)$, $\id{jump}(e)$ is recursively defined as follows:
\begin{enumerate}
\item $\id{jump}(e) := \id{jump}(e_1)$ if $k = 1$ (thus $\esuf(e)=( e_1 )$), and 
\item $\id{jump}(e) := (e_1, \ldots, e_k)$ if $k \ge 2$.

\end{enumerate}
\end{definition}

Note that $\lab(e)$ equals $\lab(e_1) \cdots \lab(e_k)$ for $jump(e)=(e_1, \ldots, e_k)$.
\begin{lemma} \label{lem:jump:termination}
  For any edge $e$ with $slen(e)\ge 2$, $\id{jump}(e)$ consists of at least two edges.
\end{lemma}

\begin{proof}
  Assume on the contrary that
  $\id{jump}(e) = e'$ for some edge $e'$.
  This implies $\slen(e') \geq 2$.
  By definition, $e'.hi$ is a proper suffix of $e.hi$,
  namely, there exists an integer $k \geq 1$ such that
  $\suflink^k(e.hi) = e'.hi$.
  For any character $c$ which appears in $T$,
  there is a (type-1 or type-2) node
  which represents $c$ as a child of the source of $\LCDAWG(T)$.
  This implies that there is an out-going edge $e''$ of length $1$
  from the source representing the first character of $e.hi$.
  This contradicts that $\id{jump}(e)$ only
  contains a single edge $e'$ with $\slen(e') \geq 2$.
  \qed
\end{proof}

\begin{theorem}
  \label{lemma:jump:time}  
   For a given $\LCDAWG(T)$, there is an algorithm that computes all jump links in $O(\tilde e_T \log \sigma)$ time. 
\end{theorem}

\begin{proof}
  \newcommand{\xx}{Y}
  
  We explain how to obtain $jump(e)$ for an edge $e$ with $slen(e) \geq 2$.
  For all edge $e$ with $slen(e) \geq 2$, we manage a pointer to the first edge $e^\prime$ of $jump(e)$ by $P[e]=e^\prime$.
  We initially set $P[e]=\epsilon$ for all $e$.
  For all nodes $e$ with $slen(e) \geq 2$, let $u$ be an outgoing edge of $\slink(e.hi)$ with the same label character of $e$.
  We check whether $P[e]=\epsilon$ and, if so, we recursively compute $P[u]$, and then set $P[e]=P[u]$.
  In this way all $P[e]$ can be computed in $O(\tilde e_T \log \sigma)$ time in total, where the $\log \sigma$ is needed for selecting the out going edge. 
  From Lemma~\ref{lem:esuf:nonbranch}, since there does not exist branching edge on each jump link, $jump(e)$ can be easily obtained from $P[e]$ by traversing the path until encountered a type-1 node. \qed
\end{proof}
  
\mysubsubsection{An SLP for the L-CDAWG}
We build an SLP which represents all edge labels in $\LCDAWG(T) = (V, E)$
based on jump links.
For each edge $e$, let $X(e)$ denote
the variable which generates the string label $\lab(e)$.
Let $E = \{e_1, \ldots, e_s\}$.
For any $e_i \in E$ with $\slen(e_i) = 1$,
we construct a production $X(e_i) \to c$
where $c \in \Sigma$ is the label.
For any $e_i \in E$ with $\slen(e_i) \geq 2$,
let $\id{jump}(e_i) = (e'_1, \dots, e'_k)$.
We construct productions
$X(e_i) \to X(e'_1)Y_1$, $Y_1 \to X(e'_2) Y_2$, \ldots, $Y_{k-3} \to X(e'_{k-2})Y_{k-2}$, and $Y_{k-2} \to X(e'_{k-1}) X(e'_{k})$.
We call a production whose left-hand size is $Y_{i}$
an \emph{intermediate} production.
It is clear that $X(e_i)$ generates $\lab(e)$
and we introduced $k-1$ productions.
If there is another edge $e_j$~($i \neq j$)
such that $\id{jump}(e_j) = (e'_1, \dots, e'_k)$,
then we construct a new production $X(e_j) \to X(e'_1)Y_1$
and reuse the other productions.
Let $p$ be the path that spells out the text $T$.
We create productions which generates $T$
using the same technique as above for this path $p$.
Overall, the total number of
intermediate productions
is linear in the number of type-2 nodes in $\LCDAWG(T)$.
Since there are $O(|E|)$ non-intermediate productions,
this SLP consists of $O(\tilde e_T)$ productions.

Now, we have the main result of this subsection. 

\begin{theorem}
  \label{lemma:slp:correct:time}  
  For a given $\LCDAWG(T)$, there is an algorithm that constructs an SLP 
   which represents all edge labels in $O(\tilde e_T \log \sigma)$ time. 
\end{theorem}

\begin{proof}
  \newcommand{\xx}{Y}
  By the above algorithm,
  if jump links are computed, 
  we can obtain an SLP which represents all edge labels in $O(\tilde e_T)$ time. 
  From Theorem~\ref{lemma:jump:time}, we can compute all jump links in $O(\tilde e_T \log \sigma)$ times.
  Overall, the total time of this algorithm is  $O(\tilde e_T \log \sigma)$. 
  \qed
\end{proof}

  Fig.~\ref{fig:slp_lcdawg} shows $\LSTrie(T)$ and $\LCDAWG(T)$
enhanced with the SLP for string $T = abcdbcda\$$.

We associate to each edge label the corresponding variable of the SLP. 
By applying algorithms of Gasieniec \emph{et al.}~\cite{Gasieniec:kolpakov:DCC:2005} (in Proposition~\ref{Gasieniec}) and Bille \emph{et al.}~\cite{Bille:Landau:Raman:Sadakane:Rao:Weimann:2015} (in Proposition~\ref{Bille}), we can show the following theorems.

\begin{theorem} 
  For a text $T$,  $\LCDAWG(T)$ can support 
  pattern matching for a pattern $P$ of length $m$ in $O(m \log \sigma + occ)$ time. 
\end{theorem} 

\begin{proof}
  From Proposition~\ref{Gasieniec},  
  any consecutive $m$ characters from the beginning of an edge  in $\LCDAWG(T)$ can be sequentially read in $O(m)$ time.
  $\CDAWG(T)$ can support pattern matching by traversing the path from the source with $P$ in $O(m \log \sigma + occ)$ time~\cite{blumer:bbhme:1987complete}.
  Since $\LCDAWG(T)$ contains the topology of  $\CDAWG(T)$, 
  it can also support pattern matching in $O(m \log \sigma + occ)$ time.
  \qed
\end{proof}

\begin{theorem} 
  For a text $T$ of length $n$,
  $\LCDAWG(T)$ has an SLP that derives $T$. 
  In addition, we can read any substring $T[i..i + m]$ can be read in $O(m + \log n)$ time. 
\end{theorem} 

\begin{proof}
  The text $T$ of $\LCDAWG(T)$ is represented by the longest path $p$ from the source to the sink. 
  Remembering $p$ makes it possible to read any position of $T$ by using the Proposition~\ref{Bille}.
  \qed 
\end{proof}

\subsection{The main result}
\label{subsec:const}

It is known that for a given string $T$ of length $n$
over an integer alphabet of size $n^{O(1)}$,
$\CDAWG(T)$ can be constructed in $O(n)$ time~\cite{Narisawa:Inenaga:Bannai:Takeda:Algorithmica}.
Combining this with the preceding discussions,
we obtain the main result of this paper.

\begin{theorem}
  For a text $T$ of length $n$, 
  $\LCDAWG(T)$ 
  supports pattern matching in $O(m \log \sigma + \occ)$ time
  for a given pattern of length $m$ and 
  substring extraction in $O(m + \log n)$ time
  for any substring of length $m$,
  and can be stored in $O(\tilde e_T \log n)$ bits of space
  (or $O(\tilde e_T)$ words of space).
  If $\CDAWG(T)$ is already constructed,
  then $\LCDAWG(T)$ can be constructed in $O(\tilde e_T\log \sigma)$
  total time.
  If $T$ is given as input,
  then $\LCDAWG(T)$ can be constructed in $O(n + \tilde e_T \log \sigma)$
  total time for integer alphabets of size $n^{O(1)}$.
  After $\LCDAWG(T)$ has been constructed,
  the input string $T$ can be discarded.
\end{theorem}

\section{Conclusions and further work}
\label{sec:conc}

In this paper, we presented a new repetition-aware data structure
called Linear-size CDAWGs.
$\LCDAWG(T)$ takes linear space in the number 
of the left- and right-extensions of the maximal repeats in $T$,
which is known to be small for highly repetitive strings.
The key idea is to introduce type-2 nodes
following LSTries proposed by
Crochemore~\emph{et al.}~\cite{epifanio:mignosi:grossi:crochemore:2016lst}.
Using a small SLP induced from edge-suffix links
that is enhanced with random access and prefix extraction data structures,
our $\LCDAWG(T)$ supports efficient pattern matching
and substring extraction.
This SLP is repetition-aware, i.e.,
its size is linear in the number of left- and right-extensions
of the maximal repeats in $T$.
We also showed how to efficiently construct $\LCDAWG(T)$.

Our future work includes implementation of $\LCDAWG(T)$ and
evaluation of its practical efficiency, when compared with
previous compressed indexes for repetitive texts.
An interesting open question is whether we can efficiently construct
$\LCDAWG(T)$ in an \emph{on-line manner} for growing text.



\end{document}